\documentclass[12pt,draftcls,onecolumn]{IEEEtran}
\ifCLASSINFOpdf
\else
\fi
\usepackage{graphicx}

\newtheorem{thm}{Theorem}[section]

\newtheorem{rem}{Remark}[section]
\begin{document}
%
\title{Enabling Differentiated Services Using Generalized  Power Control Model in Optical Networks}

\author{\IEEEauthorblockN{Quanyan Zhu}
\IEEEauthorblockA{Department of Electrical and Computer Engineering\\
University of Toronto, Ontario M5S 3L1\\
Email: qzhu@control.utoronto.ca}
\and
\IEEEauthorblockN{Lacra Pavel}
\IEEEauthorblockA{Department of Electrical and Computer Engineering\\
University of Toronto, Ontario M5S 3L1\\
Email: pavel@control.utoronto.ca}}
\author{Quanyan~Zhu,~\IEEEmembership{Student Member,~IEEE,}
        Lacra~Pavel,~\IEEEmembership{Senior Member,~IEEE}
\thanks{Quanyan Zhu is with the Department of Electrical and Computer Engineering, University of Illinois at Urbana Champaign, IL, 61801, USA email: zhu31@illinois.edu; L. Pavel is with the Department
of Electrical and Computer Engineering, University of Toronto, Toronto,
ON, M5S 3L1 Canada e-mail:pavel@control.utoronto.ca.}
}

%


\maketitle
\begin{abstract}
This paper considers a generalized framework to study OSNR optimization-based end-to-end link level power control problems in optical networks. We combine favorable features of game-theoretical approach and central cost approach to allow different service groups within the network. We develop solutions concepts for both cases of empty and nonempty feasible sets. In addition, we derive and prove the convergence of a distributed iterative algorithm for different classes of users. In the end, we use numerical examples  to illustrate the novel framework.
\end{abstract}
%
\IEEEpeerreviewmaketitle
\section{Introduction}

Reconfigurable optical Wavelength-Division Multiplexing (WDM) communication networks
with arbitrary topologies are currently enabled by technological advances in optical devices such as optical add/drop MUXes (OADM), optical cross connects (OXC) and dynamic gain equalizer (DGE). It is important that channel transmission performance and quality of service (QoS) be optimized and
maintained after reconfiguration. At the
physical transmission level, channel performance and QoS
are directly determined by the bit-error rate (BER), which
in turn depends on optical signal-to-noise ratio (OSNR),
dispersion and nonlinear effects, \cite{ARG05}. Thus, OSNR is considered as the dominant performance parameter in link-level optimization. Conventional off-line OSNR optimization is done by adjusting channel input
power at transmitter (Tx) to equalize the dominant impairment of noise accumulation in chains of optical amplifiers.  However, for reconfigurable optical networks, where
different channels can travel via different optical paths, it is more desirable to implement on-line
decentralized iterative algorithms to accomplish such adjustment.

Recently, this problem is addressed in many research works \cite{PAV06b},\cite{PAV06a},\cite{PAN05}, and two
optimization-based approaches are prevalently used: the central cost and the non-cooperative game approach.
The goals and models of the two approaches are inherently
different. Central cost approach satisfies the  target
OSNR with minimum total power consumption. The model embeds the OSNR requirements in its constraints and indirectly optimizes a certain design criterion. Such model yields a relatively simple closed-form solution; however, it doesn't optimize OSNR in a direct fashion, and thus, channel performance can be potentially improved for users who need higher quality of transmission. On the other hand, the game approach is a naturally distributed model which directly optimizes OSNR based on a payoff function in a non-cooperative manner. Each user optimizes her own utility to achieve the best possible OSNR. The solution from this approach is given by Nash equilibrium. As a result, this solution concept yields best achievable OSNR levels for each user. Since the game approach involves a cost function arising from pricing, it gives an over-allocation of resources. Some users may wish to avoid such cost and only demand a basic level of transmission. Apparently, these two approaches are for two different type of users and different transmission purposes.

To make use of  the advantages from each approach, we propose a generalized model that combines their features. Such a generalization allows to accommodate different types of users and also provides a novel mixed framework to study OSNR power control problem. We separate users into two different categories. One type of users are those who are willing to pay a price to fully optimize their transmission performance. Another type of users are those who are content with basic transmission quality, or OSNR level, set by the network. The quality of service (QoS) can be met for the former by a game-theoretically based optimization approach; and for the later by a mechanism similar to central cost approach.

The contribution of this paper lies in the capability of service differentiation of the generalized model. For simplicity, total capacity constraints are not considered. The paper is organized as follows. In section 2, we review the network OSNR  model and the basic concepts about the two optimization-based approaches. In section 3, we establish a general framework and propose two solution concepts for two different cases of feasible sets. Section 4 gives an iterative algorithm to achieve such solutions in the framework. This is illustrated in section 5 by numerical examples. Section 6 concludes the paper and points out future directions of research.

\section{Background}

\subsection{Review of Optical Network Model}

Consider a network with a set of optical links
$\mathcal{L}=\{1,2,..,L\}$ connecting the optical nodes, where
channel add/drop is realized. A set $\mathcal{N}=\{1,2,...,N\}$
of channels are transmitted, corresponding to a set of
multiplexed wavelengths. Illustrated in Figure \ref{oplink}, a link $l$ has $K_l$ cascaded optically
amplified spans. Let $N_l$ be the set of channels transmitted
over link $l$. For a channel $i \in \mathcal{N}$, we denote by
$\mathcal{R}_i$ its optical path, or collection of links, from
source (Tx) to destination (Rx). Let $u_i$ be the $i$th channel
input optical power (at Tx), and $\textbf{u}=[u_1,...,u_N]^T$
the vector of all channels' input powers. Let $s_i$ be the
$i$th channel output power (at Rx), and $n_i$ the optical noise
power in the $i$th channel bandwidth at Rx. The $i$th channel
optical OSNR is defined as $OSNR_i=\frac{s_i}{n_i}$. In~\cite{PAV06b}, some assumptions are made to simplify the
expression for OSNR, typically for uniformly designed optical
links:
\begin{enumerate}
\item (A1) Amplified spontaneous emission (ASE) noise power does not participate in amplifier
    gain saturation.
\item (A2) All the amplifiers in a link have the same
    spectral shape with the same total power target and are operated in automatic power control mode.
\end{enumerate}
Under A1 and A2, dispersion and nonlinearity are considered to be limited, and ASE noise accumulation will be the dominant impairment. The OSNR for the $i$th channel is given as
\begin{equation}\label{OSNR}
    OSNR_i=\frac{u_i}{n_{0,i}+\sum_{j\in
\mathcal{N}}\Gamma_{i,j}u_j},  i\in\mathcal{N}
\end{equation}
where $\mathbf{\Gamma}$ is the full $n \times n$ system matrix
which characterizes the coupling between channels. $n_{0,i}$
denotes the $i$th channel noise power at the transmitter.
System matrix $\mathbf{\Gamma}$ encapsulates the basic physics
present in optical fiber transmission and implements an abstraction from a network to an
input-output system. This approach has been used in~\cite{SAAB02} for the wireless case to model CDMA uplink communication.
Different from the system matrix used in
wireless case, the matrix
$\mathbf{\Gamma}$ given in (\ref{eqnGamma}) is commonly
asymmetric and is more complicatedly dependent on parameters
such as spontaneous emission noise, wavelength-dependent gain,
and the path channels take.
\begin{equation}\label{eqnGamma}
  \Gamma_{i,j}=\sum_{i\in\mathcal{R}_i}\sum_{k=1}^{K_l}\frac{G_{l,j}^{k}}{G_{l,i}^{k}}
\left(\prod_{q=1}^{l-1}\frac{\mathbf{T}_{q,j}}{\mathbf{T}_{q,i}}\right)\frac{ASE_{l,k,i}}{P_{o,l}},
 \forall j\in \mathcal{N}_l.
\end{equation}
where $G_{l,k,i}$ is the wavelength dependent gain at $k$th
span in $l$th link for channel $i$;
$\mathbf{T}_{l,i}=\prod_{q=1}^{K_l}G_{l,k,i}L_{l,k}$ with
$L_{l,k}$ being the wavelength independent loss at $k$th span
in $l$th link; $ASE_{l,k,i}$ is the wavelength dependent
spontaneous emission noise; $P_{0,l}$ is the output power at each span.
\begin{figure}
\begin{center}
  \includegraphics[scale=0.3]{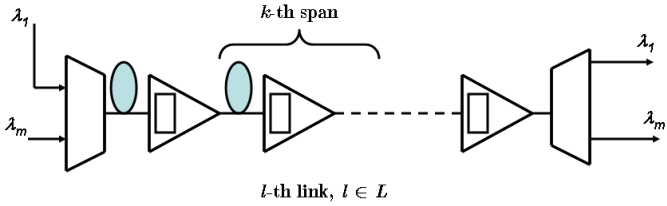}\\
  \caption{A Typical Optical Link in DWDM Optical Networks }\label{oplink}
  \end{center}
\end{figure}

\subsection{Central Cost Approach}
Similar to the SIR optimization problem in the wireless
communication networks~\cite{SMG02,SG01}, OSNR optimization achieves
the target OSNR predefined by each channel user by allowing the
minimum transmission power. Let $\gamma_i, i\in \mathcal{N}$ be
the target OSNR for each channel. By setting the OSNR
requirement as a constraint, we can arrive at the following
central cost optimization problem (CCP):
\begin{equation}\label{CCAOpt}
\begin{tabular}{ll}
  $\textrm{(CCP) }$&$\min_\mathbf{u} \sum_{i\in \mathcal{N}}u_i$ \\
  $\textrm{subject to }$& $OSNR_i \ge \gamma_i \texttt{ } \forall
i\in \mathcal{N}.$ \\
\end{tabular}
\end{equation}
Under certain conditions, it has been shown in~\cite{PAV06b} that the feasible set of (CCP) is nonempty and the optimal solution is achievable at the boundary of the feasible set.

The formulated optimization problem can be extended to
incorporate more constraints such as
\begin{equation}\label{CCAConst3}
     u_{i,\min} \leq u_i
\leq u_{i,\max},
\end{equation}
where $u_{i,\min}$ is minimum threshold power required for
transmission for channel $i$ and  $u_{i,\max}$ is maximum power
channel $i$ can attain. In the central cost approach, power
$u_i$ are the parameters to be minimized and the objective
function is linearly separable. In addition, the constraints are
linearly coupled. These nice characteristics in central cost
approach leads to a relatively simple optimization problem.

\subsection{Non-cooperative Game Approach}
Let's review the basic game-theoretical model for power control in
optical networks without constraints. Consider a game defined by a triplet $\langle
\mathcal{N},(A_i),(J_i)\rangle$. $\mathcal{N}$ is the index set of
players or channels; $A_i$ is the strategy set $\{u_i\mid u_i\in
[u_{i,\min},u_{i,\max}]\}$; and, $J_i$ is the cost function. It is
chosen in a way that minimizing the cost is related to maximizing
OSNR level.  In~\cite{PAV06a}, $J_i$ is defined as
\begin{equation}\label{GAUtil}
    J_i(u_i,u_{-i})=\alpha_i u_i - \beta_i \ln
\left(1+a_i\frac{u_i}{X_{-i}}\right),  i\in\mathcal{N}
\end{equation}
where $\alpha_i,\beta_i$ are channel specific parameters, that
quantify the willingness to pay the price and the desire to maximize
its OSNR, respectively, $a_i$ is a channel specific parameter,
 $X_{-i}$ is defined as $X_{-i}=\sum_{j \neq
i} \Gamma_{i,j}u_j+n_{0,i}$. This specific choice of utility
function is non-separable, nonlinear and coupled. However, $J_i$ is
strictly convex in $u_i$ and takes a specially designed form such
that its first-order derivative is linear with respect to
$\mathbf{u}$.

The solution from the game approach is usually characterized by
Nash equilibrium (NE). Provided that $\sum_{j\neq i}\Gamma_{i,j}< a_i$, the
resulting NE solution is uniquely determined in a closed form by
\begin{equation}\label{GASoln}
    \mathbf{\widetilde{{\Gamma}} u^*=\widetilde{b}},
\end{equation}
where $\widetilde{\Gamma}_{i,j}=a_i,$ for
$j=i$; $\widetilde{\Gamma}_{i,j}=\Gamma_{i,j},$ for $j \neq i$
and $\widetilde{b}=\frac{a_i\beta_i}{\alpha_i}-n_{0,i}$.

Similar to the wireless case \cite{SAAB02}, we are able to
construct iterative algorithms to achieve the Nash equilibrium.
A simple deterministic first order parallel update algorithm is:
\begin{equation}\label{GAAlg2}
    u_i(n+1)=\frac{\beta_i}{\alpha_i}-\frac{1}{a_i}\left(\frac{1}{OSNR_i(n)}-\Gamma_{i,i}\right)u_i(n).
\end{equation}
As proved in \cite{PAV06a}, the algorithm (\ref{GAAlg2}) converges to
Nash equilibrium $\mathbf{u}^*$ provided that
$\frac{1}{a_i}\sum_{j\neq i}\Gamma_{i,j} < 1, \forall i$.

\section{Generalized Model}

In this section, we consider a game designed to allow service differentiation by separating users into two groups: one group seeking a minimum OSNR target and another group participating in a game setting for OSNR optimization. The minimum OSNR for target seekers is set by the network to ensure the minimum quality of service. However, the game players can submit their parameters and optimize their service accordingly, but they have to pay a price set by the network for unit power consumption. This concept is illustrated in Figure \ref{GPTS}.  Let's denote set $\mathcal{N}_1=\{1,2,...,N_1\}$ as the set of competitors, i.e. users who wish to compete for an optimal OSNR. Let set $\mathcal{N}_2=\{N_1+1,\cdots,N_2\}$ be the group of users with target OSNR given by $\gamma_i, i\in \mathcal{N}_2$. Let $\mathcal{N}=\mathcal{N}_1\cup\mathcal{N}_2$, $m=|\mathcal{N}_1|=N_1$, $n=|\mathcal{N}_2|$, $N=|\mathcal{N}|=m+n$ and $\mathbf{u}=[u_1,\cdots,u_{N_1},u_{N_1+1},\cdots,u_{N_2}]^T$.

\begin{figure}
\begin{center}
  \includegraphics[scale=0.35]{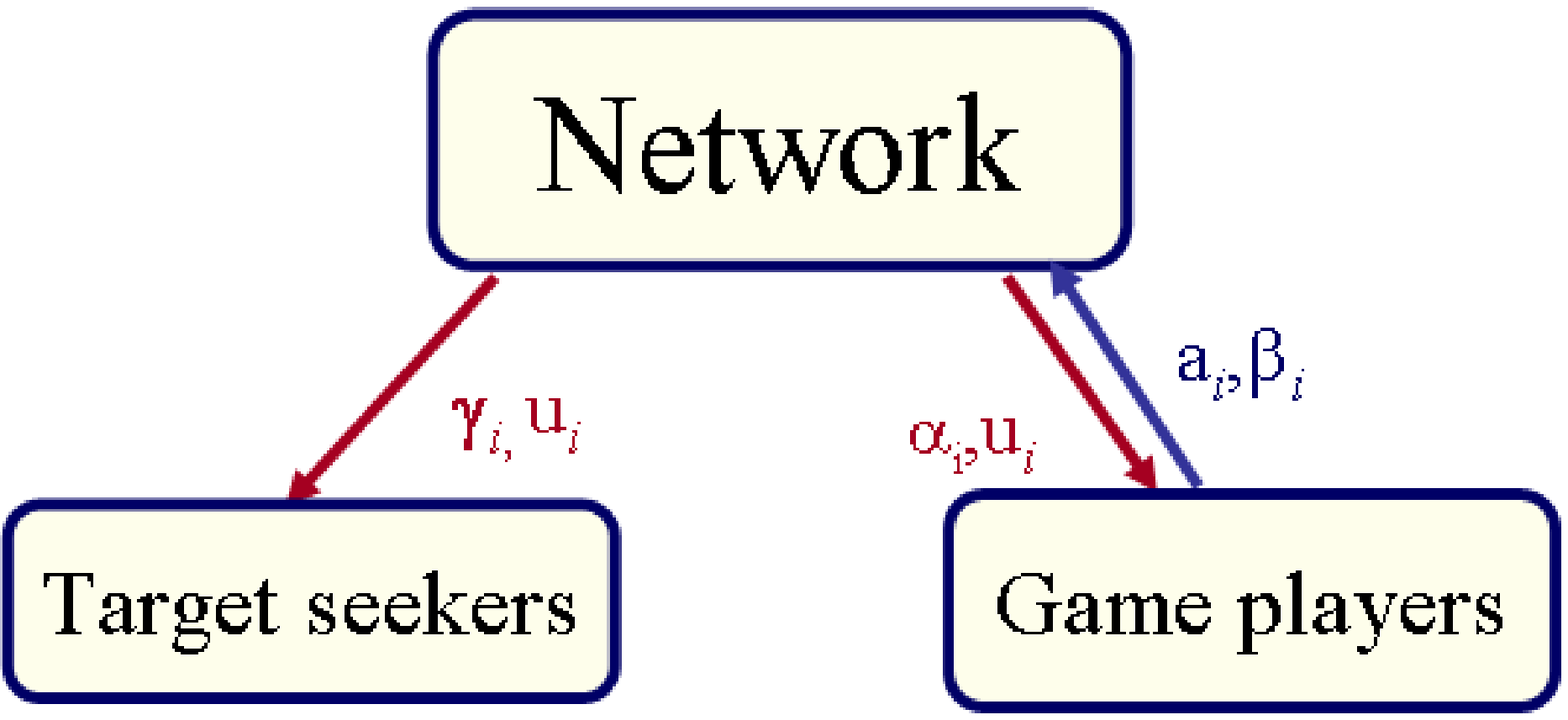}\\
  \caption{Game players and target seekers in the network}\label{GPTS}
  \end{center}
\end{figure}

For the game-theoretical players, using the cost function given in (\ref{GAUtil}), we can form a system of equations given by
$$ a_iu_i+X_{-i}=\frac{a_i\beta_i}{\alpha_i}, \forall i\in\mathcal{N}_1 $$ and thus,
$\widetilde{\mathbf{\Gamma}}\mathbf{u}=\widetilde{\mathbf{b}},$ where $\widetilde{\mathbf{\Gamma}}\in \mathcal{R}^{m\times N}$ and $\widetilde{\mathbf{b}}\in \mathcal{R}^{m}$ are defined as in (\ref{GASoln}). Users with target OSNR shall have $\mathbf{u}$ satisfy
$OSNR_i\geq \gamma_i, \forall i\in \mathcal{N}_2,$
or equivalently from (\ref{OSNR}),
$$\frac{u_i}{\Gamma_{i,i}u_i+\sum_{j\neq i}\Gamma_{i,j}u_j+n_{0,i}}\geq \gamma_i$$
and thus in a matrix form,
$\widehat{\mathbf{\Gamma}}\mathbf{u}\geq\widehat{\mathbf{b}},$
where $\widehat{\mathbf{b}}=[\gamma_1n_{0,1},\cdots,\gamma_Nn_{0,N}]^T\in \mathcal{R}^{n}$, $\widehat{\mathbf{\Gamma}}\in \mathcal{R}^{n\times N}$ and is given in (\ref{GammaN2}).
\begin{figure*}
\begin{equation}\label{GammaN2}
\widehat{\mathbf{\Gamma}}=\left[
  \begin{array}{ccccc}
    -\gamma_{N_1+1}\Gamma_{N_1+1,1} & \cdots & 1-\gamma_{N_1+1}\Gamma_{N_1+1,N_1+1} & \cdots & -\gamma_{N_1+1}\Gamma_{N_1+1,N} \\
    \vdots & \ddots & \ddots & \ddots & \vdots \\
    -\gamma_{N-1}\Gamma_{N-1,1} & -\gamma_{N-1}\Gamma_{N-1,2} & \cdots & 1-\gamma_{N-1}\Gamma_{N-1,N-1} & -\gamma_{N-1}\Gamma_N \\
    -\gamma_N\Gamma_{N,1} & -\gamma_N\Gamma_{N,2} & \cdots & \cdots & 1-\gamma_N\Gamma_{N,N} \\
  \end{array}
\right].
\end{equation}
\hrulefill
\vspace*{4pt}
\end{figure*}
Let $F_1=\{\mathbf{u}\in\mathcal{R}^N\mid \widetilde{\mathbf\Gamma}\mathbf{u}=\widetilde{\mathbf{b}}\}$ and
$F_2=\{\mathbf{u}\in\mathcal{R}^N\mid\widehat{\mathbf{\Gamma}}\mathbf{u}\geq\widehat\mathbf{b}\}$.
In summary, we have a problem  formulated as in (DS), where we find solutions that satisfy $F_1$ subject to the constraint described by $F_2$.
\begin{equation}\label{DS1}
\begin{array}{cc}
  \textrm{(DS)} & \widetilde{\mathbf{\Gamma}}\mathbf{u}=\widetilde{\mathbf{b}}\\
  \textrm{s.t.} &\widehat{\mathbf{\Gamma}}\mathbf{u}\geq\widehat{\mathbf{b}}   \end{array}
\end{equation}

In the following discussion, we separate (DS) into two cases: (1) $F=F_1\cap F_2\neq\emptyset$, (2)$F=F_1\cap F_2=\emptyset$, which require different techniques to find appropriate solutions.

\subsection{Non-empty Feasible Set}

A non-empty $F$ may give rise to multiple points that solve (DS). We may impose some design criteria, or, objective function to reformulate DS for finding an appropriate solution that solves DS and meet the design criteria at the same time.

We can use the following result to ensure the nonempty feasible set $F$.
\begin{thm}\label{nonemptyF}
If $\overline\mathbf{\Gamma}=\left[
                               \begin{array}{c}
                                 \widetilde\mathbf{\Gamma} \\
                                 \widehat\mathbf{\Gamma} \\
                               \end{array}
                             \right]
$ is nonsingular, the feasible set $F=F_1\cap F_2$ is non-empty.
\end{thm}
\begin{proof}
Let $\mu\in\mathcal{R}^n_+$ a nonnegative vector. Equivalently, we can express $F_2$ into $F_2=\{\mathbf{u}\in\mathcal{R}^n\mid \widehat\mathbf\Gamma \mathbf{u}=\widehat\mathbf{b}+\mu,\textrm{~for~some~}\mu\in\mathcal{R}^n_+\}$. The set $F$ is thus equivalently
$F=\{\mathbf{u}\in\mathcal{R}^N\mid \overline{\mathbf\Gamma}\mathbf{u}=\mathbf{\phi},\textrm{~for~some~}\mu\in\mathcal{R}^n_+\}$,
where $\overline\mathbf{\Gamma}=\left[
                     \begin{array}{c}
                       \widetilde\mathbf{\Gamma} \\
                       \widehat\mathbf{\Gamma} \\
                     \end{array}
                   \right]
$ and $\mathbf\phi=\left[
                     \begin{array}{c}
                       \widetilde\mathbf{b} \\
                       \widehat\mathbf{b}+\mu \\
                     \end{array}
                   \right]
$. If $\overline\mathbf{\Gamma}$ is nonsingular, there exist a unique $\mathbf{u}\in\mathcal{R}^N$ for every nonnegative $\mu$. Therefore $F$ is non-empty.
\end{proof}

Suppose conditions in Theorem \ref{nonemptyF} hold and $F$ is nonempty. We consider an appropriate solution in $F$ that satisfies a certain design criteria. Thus, we formulate (DSNP\footnote{DSNP stands for ``Differentiated Service N-person Problem''.}) in which we minimize total power consumption subject to the conditions arising from the different service requirements.

\begin{equation}\label{DSNP}
\begin{array}{cc}
  \textrm{(DSNP)} & \min \sum_iu_i\\
   \textrm{s.t.} & \widetilde{\mathbf{\Gamma}}\mathbf{u}=\widetilde{\mathbf{b}},\widehat{\mathbf{\Gamma}}\mathbf{u}\geq\widehat{\mathbf{b}}
\end{array}
\end{equation}

The constraints of (DSNP) can be relaxed and augmented into
\begin{equation}\label{relaxedEq}
\overline\mathbf{\Gamma}\mathbf{u}\geq\overline\mathbf{b}.
\end{equation}
where
$\overline{\mathbf{\Gamma}}=\left[
                              \begin{array}{c}
                                \widetilde{\mathbf{\Gamma}} \\
                                \widehat{\mathbf{\Gamma}} \\
                              \end{array}
                            \right]
\in\mathcal{R}^{N\times N}$
and $\overline{\mathbf{b}}=\left[
                             \begin{array}{c}
                                \widetilde{\mathbf{b}} \\
                                \widehat{\mathbf{b}} \\
                             \end{array}
                           \right]
\in\mathcal{R}^{N}.$

According to the fundamental theorem of linear programming \cite{BER03}, if (DSNP) is realistic, the solution is obtained at the extreme point of the feasible set $F$. Since $F$ has only one extreme point when $\overline\mathbf\Gamma$ is non-singular, the solution is uniquely given by
\begin{equation}\label{eqSoln}
\mathbf{u}=\overline\mathbf{\Gamma}^{-1}\overline\mathbf{b}.
\end{equation}

To further characterize the solution $\mathbf{u}$, we assume strict diagonal dominance of matrix $\overline\mathbf{\Gamma}$ \cite{Horn90}, which leads to non-singularity of the matrix and uniqueness of the solution.

\begin{thm}\label{StrictDom}
Suppose OSNR targets $\gamma_i, i\in\mathcal{N}_2$ are chosen such that $\gamma_i<\frac{1}{\sum_{j\in\mathcal{N}}\Gamma_{i,j}}, i\in\mathcal{N}_2$. In addition, parameters $a_i$ are chosen as $a_i>\sum_{j\neq i,j\in\mathcal{N}}\Gamma_{ij}, \forall i\in\mathcal{N}_1.$ The matrix $\overline\mathbf{\Gamma}$ is strictly diagonally dominant. And thus, a unique solution to (DSNP) is given by (\ref{eqSoln}).
\end{thm}

\begin{proof}
From the assumption that $\gamma_i\sum_{j\in\mathcal{N}}\Gamma_{ij}<1,i\in\mathcal{N}_2$, it is apparent that $\gamma_i<\frac{1}{\Gamma_{ii}}$ and $\left|1-\gamma_i\Gamma_{ii}\right|>\gamma_i\sum_j\Gamma_{ij}, \forall i\in\mathcal{N}_2$. In addition, $a_i>\sum_{j\neq i,j\in\mathcal{N}}\Gamma_{i,j}, \forall i\in\mathcal{N}_1$. Therefore, matrix $\overline\mathbf{\Gamma}$ is strictly diagonally dominant. Using Gershgorin theorem in \cite{Horn90}, we conclude that there exists a unique solution to (DSNP).
\end{proof}

The assumption of strict diagonal dominance in Theorem \ref{StrictDom} is reasonable because typical values of $\Gamma_{ij}$ are found to be on the order of $10^{-3}$ and desirable levels of OSNR are 20-30dB.
\begin{rem}
(DSNP) can be seen as a generalized approach that combines central cost approach in \cite{PAV06b} and non-cooperative game approach in \cite{PAV06a}. When $N_1=\emptyset, N_2\neq\emptyset$, (DSNP) reduces to the central cost approach. Similarly, when $N_1\neq\emptyset, N_2=\emptyset$, (DSNP) reduces to the game-theoretical approach and the given solution is Nash equilibrium accordingly. This framework allows to study two different types of users at the same time.
\end{rem}
\begin{rem}
We illustrate a two-person (DSNP), where player 1 chooses to compete and optimize his utility and player 2 chooses to meet a certain OSNR target $\gamma_2$. We form the 2-by-2 matrix $\overline{\mathbf{\Gamma}}$ and $\overline\mathbf{b}$ as follows.
$$\overline{\mathbf{\Gamma}}=\left[
  \begin{array}{cc}
    a_1 & \Gamma_{12} \\
    -\Gamma_{21}\gamma_2 & 1-\Gamma_{22}\gamma_2 \\
  \end{array}
\right],
\overline{\mathbf{b}}=\left[
                        \begin{array}{c}
                          \frac{a_1\beta_1}{\alpha_1}-n_{0,1} \\
                          n_{0,2}\gamma_2 \\
                        \end{array}
                      \right]
$$

The feasible set $F=F_1\cap F_2$ is shown in Figure \ref{DiffServGamma2} by a dotted line. The relaxed (DSNP) has its relaxed feasible depicted in the shaded region.
The solution is given by $\mathbf{u}^*=\overline{\mathbf{\Gamma}}^{-1}\overline{\mathbf{b}},$ which is illustrated by the dark point in Figure \ref{DiffServGamma2}. $\mathbf{u}^*$ is nonnegative componentwise if network price $\alpha_1$ is set such that $s_2>\frac{n_{0,2}}{1-\Gamma_{22}}$.
\end{rem}

\begin{figure}
\begin{center}
  \includegraphics[scale=0.35]{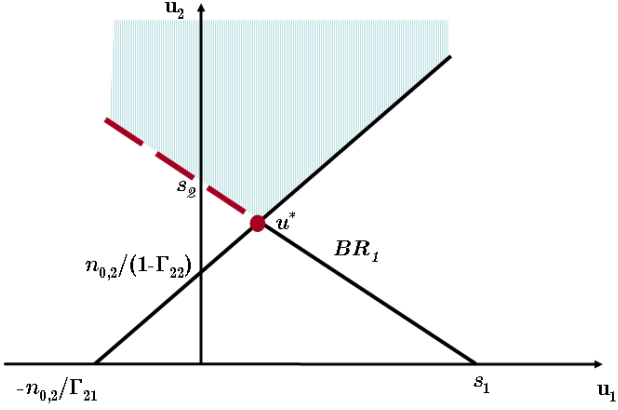}\\
  \caption{The feasible set of two-person (DSNP). $s_1=\frac{\tilde{b}_1}{a_1}$; $s_2=\frac{\tilde{b}_1}{\Gamma_{12}}$}\label{DiffServGamma2}
  \end{center}
\end{figure}

Based on Theorem \ref{StrictDom}, we can further investigate how parameters chosen by game players and target seekers influence the outcome of the allocation. The result is summarized in Theorem \ref{varControl}.

\begin{thm}\label{varControl}
Let $\kappa$ be the condition number of $\overline{\mathbf{\Gamma}}$, $T_i=a_i+\sum_{j\neq i,j\in\mathcal{N}}\Gamma_{ij}, \forall i\in\mathcal{N}_1$ and $S_k=2-2\gamma_k\Gamma_{kk}, \forall k\in\mathcal{N}_2$. Suppose $\overline\mathbf{\Gamma}$ is strictly diagonally dominant by satisfying conditions in Theorem \ref{StrictDom}. In addition, $T_i>S_k$ and $\tilde{b}_i>\hat{b}_k, \forall i\in\mathcal{N}_1, \forall k\in\mathcal{N}_2.$ The maximum allocated power allocated to users are bound as follows.
$$\frac{\max_{i\in\mathcal{N}_2}\gamma_in_{0,i}}{\max_{i\in\mathcal{N}_1}2a_i}\leq\|\mathbf{u}\|_\infty\leq\kappa\max_{i\in\mathcal{N}_1}\frac{\beta_i}{\alpha_i}$$
\end{thm}
\begin{proof}
Let $R_i$ denote the i-th row absolute sum of matrix $\overline\mathbf{\Gamma}$, i.e.,
\begin{equation}\label{Row}
R_i=\sum_{j\in\mathcal{N}}\left|\overline{\Gamma}_{ij}\right|.
\end{equation}
Using conditions from Theorem \ref{StrictDom}, we arrive at
\begin{equation}\label{Ri}
R_i=\left\{
  \begin{array}{ll}
    1+\gamma_i\sum_{j\in\mathcal{N}}\Gamma_{ij}-2\gamma_i\Gamma_{ii}<2-2\gamma_i\Gamma_{ii}, & {i\in\mathcal{N}_2;} \\
    a_i+\sum_{j\neq i,j\in\mathcal{N}}\Gamma_{ij}<2a_i., & {i\in\mathcal{N}_1.}
  \end{array}
\right.
\end{equation}

With the assumption that $a_i+\sum_{j\neq i}\Gamma_{ij}> 2-2\gamma_k\Gamma_{kk}, \forall i\in\mathcal{N}_1, \forall k\in\mathcal{N}_2,$ we obtain
$\|\overline\mathbf{\Gamma}\|_\infty=\max_{i\in\mathcal{N}}R_i=\max_{i\in\mathcal{N}}a_i+\sum_{j\neq i}\Gamma_{ij}.$
Using (\ref{Ri}) and the fact that $\Gamma_{ij}\geq0$, we obtain an upper and lower bound on $\|\overline{\mathbf{\Gamma}}\|_\infty$, i.e.,
\begin{equation}\label{ublbOGammaNorm}
  \max_{i\in\mathcal{N}_1}a_i\leq \|\overline\mathbf{\Gamma}\|_\infty \leq\max_{i\in\mathcal{N}_1}2a_i.
\end{equation}

In addition, from $\tilde{b}_i>\hat{b}_k, \forall i\in\mathcal{N}_1, \forall k\in\mathcal{N}_2,$ we obtain an upper bound and lower bound for $\|\overline\mathbf{b}\|_\infty$, given by
\begin{equation}\label{ublbObnorm}
   \max_{i\in\mathcal{N}_2}\gamma_in_{0,i}\leq\|\overline\mathbf{b}\|_\infty=\max_{i\in\mathcal{N}}\overline{b}_i\leq\max_{i\in\mathcal{N}_1}\tilde{b}_i
=\max_{i\in\mathcal{N}_1}\frac{a_i\beta_i}{\alpha_i}
\end{equation}
Since $\overline{\mathbf{\Gamma}}$ is strictly diagonally dominant, using matrix norm sub-multiplicativity, we obtain from  (\ref{eqSoln})
\begin{equation}\label{ublbu}
    \frac{\|\overline{\mathbf{b}}\|_\infty}{\|\overline{\mathbf{\Gamma}}\|_\infty}\leq\|\mathbf{u}\|_\infty\leq\frac{\kappa\|\overline{\mathbf{b}}\|_\infty}{\|\overline\mathbf{\Gamma}\|_\infty},
\end{equation}
where $\kappa$ is the condition number of $\overline{\mathbf{\Gamma}}$ given by $\kappa=\|\overline{\mathbf{\Gamma}}\|_\infty\|\overline{\mathbf{\Gamma}}^{-1}\|_\infty\geq1.$

Using (\ref{ublbOGammaNorm}), (\ref{ublbObnorm}) and (\ref{ublbu}), we obtain
\begin{eqnarray}\label{ublbu2}
\nonumber \frac{\max_{i\in\mathcal{N}_2}\gamma_in_{0,i}}{\max_{i\in\mathcal{N}_1}2a_i}\leq\|\mathbf{u}\|_\infty
 &\leq&\frac{\kappa\max_{i\in\mathcal{N}_1}a_i\beta_i/\alpha_i}{\max_{i\in\mathcal{N}_1}a_i}\\
\nonumber &\leq&\frac{\kappa\max_{i\in\mathcal{N}_1}a_i\max_{i\in\mathcal{N}_1}\beta_i/\alpha_i}{\max_{i\in\mathcal{N}_1}a_i}\\
&\leq&\kappa\max_{i\in\mathcal{N}_1}\frac{\beta_i}{\alpha_i}.
\end{eqnarray}
\end{proof}

It is easy to observe that the upper bound is dependent on the parameters of the game players and the lower bound is dependent on the OSNR levels of target seeker and parameter $a_i$ of the game players. In essence, game players control the outcome of the model and the choice of OSNR target can only affect the lower bound. Such relation describes a fair scenario in which game players, who pay for their power at $\alpha_i$, have their choices of parameters $a_i,\beta_i$ to influence the network allocation.

\begin{rem}
Since $\|\mathbf{u}\|_\infty\leq\|\mathbf{u}\|_2\leq \sqrt{N}\|\mathbf{u}\|_\infty$, we can translate the result obtained in (\ref{ublbu2}) directly into Euclidean norm, i.e.,
\begin{equation}\label{ublbu2Euclid}
B_\infty^L\leq\|\mathbf{u}\|_2\leq \sqrt{N}B^U_\infty
\end{equation}
where $B_\infty^U=\kappa\max_{i\in\mathcal{N}_1}\frac{\beta_i}{\alpha_i}$ and $B_\infty^L=\frac{\max_{i\in\mathcal{N}_2}\gamma_in_{0,i}}{\max_{i\in\mathcal{N}_1}2a_i}$.
\end{rem}
By (\ref{ublbu2Euclid}), we can see that the network can encourage uniform channel power distribution by letting $B_\infty^U$ close to $\sqrt{N}B_\infty^L$ and provide incentive for differentiated services by letting them far apart. It can be implemented by the network by adjusting OSNR level $\gamma_i$ and pricing $\alpha_i$. Decreasing $\alpha_i$  encourages more users to be game players, giving rise to more competitions or service differentiation as a result of higher upper bound. On the other hand, increasing $\gamma_i$ raises the lower bound and encourages more users being target-seekers.

\subsection{Empty Feasible Set}
In this section, we consider the second case where feasible set $F$ is empty. Instead of finding an appropriate feasible solution, we find the closest points between set $F_1$ and $F_2$. We use a quadratic program (DS2) to minimize the error norm subject to the constraint described by $F_2$.

\begin{equation}\label{DS2}
\begin{array}{cc}
  \textrm{(DS2)} & \min_{\mathbf{u}}\|\widetilde{\mathbf{\Gamma}}\mathbf{u}-\widetilde{\mathbf{b}}\|_2\\
  \textrm{s.t.} &\widehat{\mathbf{\Gamma}}\mathbf{u}\geq\widehat{\mathbf{b}}   \end{array}
\end{equation}

We can turn the constrained problem (\ref{DS2}) into an unconstrained problem by studying its corresponding dual problem. Since $\|\widetilde{\mathbf{\Gamma}}\mathbf{u}-\widetilde{\mathbf{b}}\|_2
=\mathbf{u}^T\widetilde{\mathbf{\Gamma}}^T\widetilde{\mathbf{\Gamma}}\mathbf{u}
-2(\widetilde{\mathbf{b}}^T\widetilde{\mathbf{\Gamma}})\mathbf{u}+\widetilde{\mathbf{b}}^T\widetilde{\mathbf{b}}
$, we denote $\mathbf{H}=\frac{1}{2}\widetilde\mathbf{\Gamma}^T\widetilde\mathbf{\Gamma}, \mathbf{d}=-2(\widetilde\mathbf{\Gamma}^T\widetilde\mathbf{b})$, $\mathbf{D}=-\widehat\mathbf{\Gamma}(\mathbf{H}^T\mathbf{H})^{-1}\mathbf{H}^T\widehat\mathbf\Gamma^{T}$, $\mathbf{c}=\widehat\mathbf{b}+\widehat\mathbf{\Gamma}(\mathbf{H}^T\mathbf{H})^{-1}\mathbf{H}^T\mathbf{d}$; and form a Lagrangian from the original problem (DS2).
\begin{eqnarray}\label{dualL}
 D(\mu) &=& \min_{\mathbf{u}}\mathcal{L}(\mathbf{u},\mu) \\
\nonumber    &=& \min_{\mathbf{u}}\left(
    \frac{1}{2}\mathbf{u}^T\mathbf{H}\mathbf{u}+\mathbf{d}^T\mathbf{u}+\widetilde\mathbf{b}^T\widetilde\mathbf{b}
    +\mu^T(-\widehat\mathbf\Gamma \mathbf{u}+\widetilde\mathbf{b})
    \right)
\end{eqnarray}

Since the objective function is convex, the necessary and sufficient
condition for a minimum is that the gradient must vanish,i.e.,
\begin{equation}\label{DualImplicit}
\mathbf{H}\mathbf{u}+\mathbf{d}-\hat\mathbf{\Gamma}^T\mathbf{\mu}=0.
\end{equation}

For $n<N$, $\widetilde{\mathbf{\Gamma}}$ is not full rank. Therefore, $\mathbf{H}$ is singular and there exist multiple solutions to (\ref{DualImplicit}). Using pseudoinverse \cite{Horn90}, we can find a solution to (\ref{DualImplicit}) given by
$$\mathbf{u}=-(\mathbf{H}^T\mathbf{H})^{-1}\mathbf{H}^T\left(\mathbf{d}-\hat\mathbf{\Gamma}^T\mu\right).$$ Thus, after replacing into (\ref{dualL}), we obtain $\mu$ as a solution to the dual problem (DDS2).
\begin{equation}\label{DDS2}
\textrm{(DDS2)} \max_{\mu\geq0} \frac{1}{2}\mu^T\mathbf{D}\mu+\mu^T\mathbf{c}-\frac{1}{2}\mathbf{d}^T(\mathbf{H}^T\mathbf{H})^{-1}\mathbf{H}^T\mathbf{d}+\mathbf{b}^T\mathbf{b}
\end{equation}

The problem (LDS2) and dual problem (DDS2) can be solved using unconstrained optimization algorithms in \cite{BAZ93}, \cite{BER03}.

\section{Iterative Algorithm}
In this section, we develop algorithm for the case of nonempty $F$ set. Let $u_i(n)$ denote the power at channel $i$ at step $n$. An iterative algorithm is given as follows.
\begin{equation}\label{GenItAlg1}
\left\{
  \begin{array}{ll}
    u_i(n+1)=\frac{\beta_i}{\alpha_i}-\frac{1}{a_i}\left(\frac{1}{OSNR_i(n)}-\Gamma_{i,i}\right)u_i(n), & \forall i\in\mathcal{N}_1; \\
     u_i(n+1)=\frac{\gamma_i}{1-\gamma_i\Gamma_{i,i}}\left(\frac{1}{OSNR_i(n)}-\Gamma_{i,i}\right)u_i(n), & \forall i\in\mathcal{N}_2.
  \end{array}
\right.
\end{equation}
\begin{thm}
Algorithm (\ref{GenItAlg1}) converges provided that $a_i>\sum_{j\neq i,j\in\mathcal{N}}\Gamma_{i,j}$ and $\gamma_i$ is chosen such that $\gamma_i<\frac{1}{\sum_{j\in\mathcal{N}}\Gamma_{i,j}}$.
\end{thm}

\begin{proof}
We use a similar approach from \cite{PAV06a} to show the convergence of (\ref{GenItAlg1}).
Let's define $e_i(n)=u_i(n)-u_i^*$, where $u_i^*$ is given in (\ref{eqSoln}). Since $\overline\mathbf{\Gamma}\mathbf{u}^*=\overline\mathbf{b}$, $\widetilde{\Gamma}_{i,i}u_i^*+\sum_{j\neq i}\widetilde{\Gamma}_{i,j}u_j^*=\tilde{b}_i$, for $i\in\mathcal{N}_1$; and, $\widehat{\Gamma}_{i,i}u_i^*+\sum_{j\neq
i}\widehat{\Gamma}_{i,j}u_j^*=\hat{b}_i$, for $i\in\mathcal{N}_2$.

Substitute the expression for $u_{i}^*$ into
$e_i(n+1)$, and we obtain
$e_i(n+1)=u_i(n+1)-u_i^*=-\frac{1}{{a_i}}\left[\sum_{j\neq i}{\Gamma}_{i,j}(u_j(n)-u_j^*)\right]$, for $i\in\mathcal{N}_1$; and $e_i(n+1)=u_i(n+1)-u_i^*=\frac{1}{1-{\Gamma}_{i,i}\gamma_i}\left[\sum_{j\neq i}{\Gamma}_{i,j}\gamma_i(u_j(n)-u_j^*)\right]$, for $i\in\mathcal{N}_2$. Let $\mathbf{e}=[e_i(n)],i\in\mathcal{N}$.
Therefore, for $i\in\mathcal{N}_1$,
\begin{eqnarray}
  |e_i(n+1)|
&=& \left|\frac{1}{a_i}\left[\sum_{j\neq
i,j\in\mathcal{N}}\Gamma_{i,j}(e_j(n))\right]\right| \\
   &\leq & \frac{1}{a_i}\sum_{j\neq i,j\in\mathcal{N}}\Gamma_{i,j}\max_{j\in\mathcal{N}}|e_j(n)|\\
   &\leq & \frac{1}{a_i}\sum_{j\neq i,j\in\mathcal{N}}\Gamma_{i,j}\|\mathbf{e}(n)\|_\infty.
\end{eqnarray}
and similarly, for $i\in\mathcal{N}_2$,
\begin{eqnarray}
\nonumber  |e_i(n+1)| &=& \left|\frac{1}{1-\Gamma_{i,i}\gamma_i}\left[\sum_{j\neq
i,j\in\mathcal{N}}\Gamma_{i,j}\gamma_i(e_j(n))\right]\right| \\
   &\leq & \frac{\gamma_i}{|1-\Gamma_{i,i}\gamma_i|}\sum_{j\neq i,j\in\mathcal{N}}\Gamma_{i,j}\max_{j\in\mathcal{N}}|e_j(n)|.\\
   &\leq & \frac{\gamma_i}{|1-\Gamma_{i,i}\gamma_i|}\sum_{j\neq i,j\in\mathcal{N}}\Gamma_{i,j}\|\mathbf{e}(n)\|_\infty.
\end{eqnarray}
Since we assumed that $a_i>\sum_{j\neq i,j\in\mathcal{N}}\Gamma_{i,j}$ and $\gamma_i$ is chosen such that $\gamma_i<\frac{1}{\sum_{j\in\mathcal{N}}\Gamma_{i,j}}\leq\frac{1}{\Gamma_{i,i}}$, we can
conclude that $\|\mathbf{e}(n)\|\rightarrow 0$ from the contraction mapping
theorem. As a result, we have $u_i(n)\rightarrow u_i^*$ as
$n\rightarrow \infty$, for $i\in\mathcal{N}$.
\end{proof}
\begin{rem}
From the proof, we note that the rate of convergence of \label{GenItAlg} is determined by $$\sigma=\max\left\{\max_{i\in\mathcal{N}_1}\frac{\sum_{j\neq i,j\in\mathcal{N}}\Gamma_{i,j}}{a_i},\max_{i\in\mathcal{N}_2}\frac{\sum_{j\neq i,j\in\mathcal{N}}\Gamma_{i,j}\gamma_i}{1-\Gamma_{i,i}\gamma_i}\right\}.$$ In addition, it is easy to observe that the OSNR target-seeking users are  algorithmically equivalent to competition seeking users by letting $\beta_i/\alpha_i=0$ and $a_i=\Gamma_{i,i}-\frac{1}{\gamma_i}$, $i\in\mathcal{N}_2$. This is because no notion of pricing is used for the OSNR target seekers and they just have a utility target to meet or equivalently optimize by letting $a_i=\Gamma_{i,i}-\frac{1}{\gamma_i}$.
\end{rem}

\section{Numerical Examples}
In this section, we illustrate the concept by a MATLAB simulation. We consider an end-to-end link described in Figure \ref{oplink} with 5 amplified spans. We assume channels are transmitted at wavelengths distributed centered around 1555nm  with channel separation of 1nm. Suppose input noise power is 0.5 percent of the input signal power. The gain profile for each amplifier is identically assumed to be parabolic as in Figure \ref{GainProfile}, which is normalized with respect to  $G_{\max}=30.0$dB. Suppose 20dB is the target OSNR level for users who just want to meet a satisfactory level of transmission. We first show the case of 3 users, in which 2 users need better quality of service and one user is simply interested in meeting 20dB as a target. From Figure \ref{3userDS}, we can observe that users who need better services reach an OSNR of 26.33dB and 29.20dB, respectively. With an appropriate choice of initial conditions, the algorithm quickly converges in 1-2 steps.  In Figure \ref{30userDS}, we similarly show the case of 30 users, in which 20 are game players and 10 are target seekers.

\begin{figure}
\begin{center}
  \includegraphics[scale=0.37]{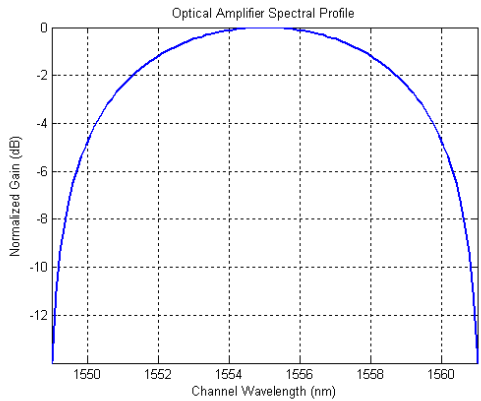}\\
  \caption{Optical Amplifier Spectral Profile}\label{GainProfile}
  \end{center}
\end{figure}

\begin{figure}
\begin{center}
  \includegraphics[scale=0.39]{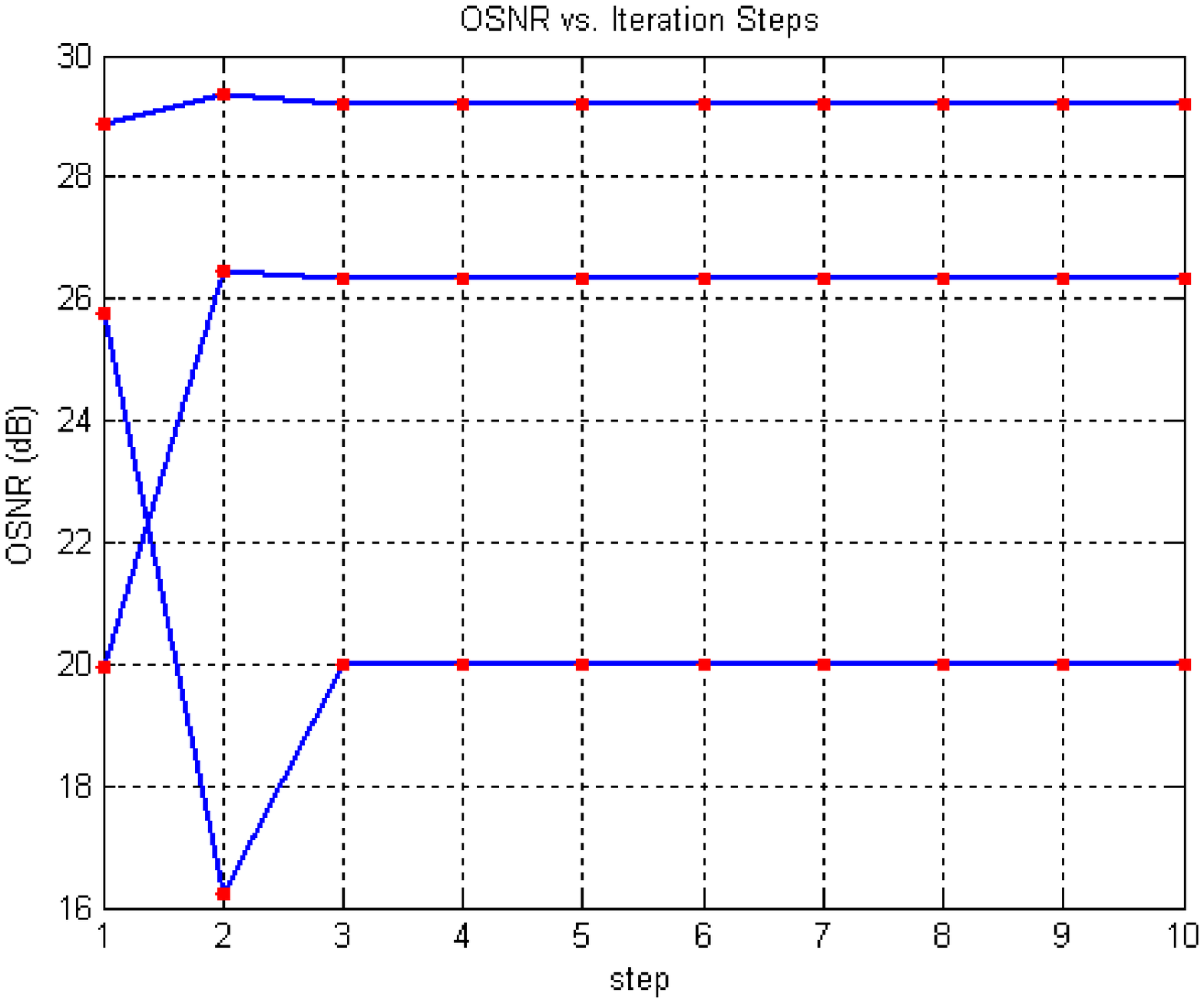}\\
  \caption{OSNR simulation with 3 users in time steps}\label{3userDS}
  \end{center}
\end{figure}

\begin{figure}
\begin{center}
  \includegraphics[scale=0.39]{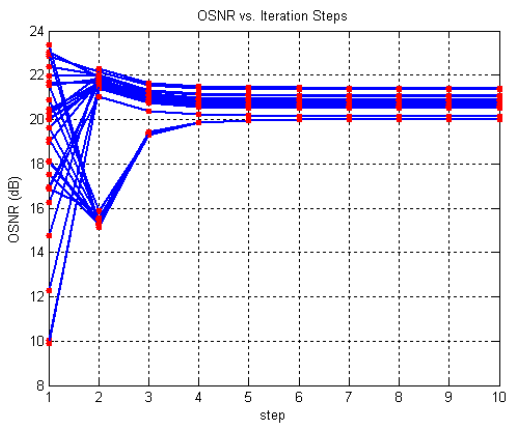}\\
  \caption{OSNR simulation with 30 users in time steps}\label{30userDS}
  \end{center}
\end{figure}

\section{Conclusion}

In this paper, we examined a generalized power control model in optical networks, which combines  features of central cost approach and game-theoretical approach. It enables two major service types in the network. One is game player, who pays for his power consumption and the other is target seeker, who is satisfied with a minimum service level set by the network. We discussed two different solutions concepts for nonempty and empty feasible set respectively and specifically designed an iterative algorithm that converges to a unique solution for the case of nonempty feasible set. The convergence of the algorithm was proved and illustrated by numerical examples of a WDM end-to-end optical link.

In this work, we didn't include capacity constraints for the sake of simplicity. We hope this work will lead to future investigations of more complicated cases where network constraints and nonlinear effects are considered. In addition, we expect this framework to be used to solve similar problems in other types of networks, for example, wireless networks.

\bibliographystyle{IEEEtran}
\bibliography{IEEEabrv,Xbib}
%

%
\end{document}